\documentclass[letterpaper, 10 pt, conference]{ieeeconf}  

\IEEEoverridecommandlockouts                              

\usepackage{amssymb}
\usepackage{amsmath}
\usepackage{graphicx}
\usepackage[colorlinks]{hyperref}
\usepackage[colorinlistoftodos]{todonotes}
\usepackage{verbatim}
\usepackage{algorithm,algcompatible,amsmath}

\newtheorem{theorem}{\bf Theorem}

\newtheorem{corollary}{\bf Corollary}

\newtheorem{lemma}{\bf Lemma}

\def\QED{~\rule[-1pt]{5pt}{5pt}\par\medskip}

\usepackage{xcolor}
\DeclareMathOperator*{\argmin}{arg\,min}

\title{\LARGE \bf
Closed-loop Parameter Identification of Linear Dynamical Systems\\ through the Lens of Feedback Channel Coding Theory
}

\author{Ali Reza Pedram$^{1}$ and Takashi Tanaka$^{2}$
\thanks{$^{1}$ Department of Mechanical          Engineering,
        University of Texas at Austin, Austin, TX, USA.
        {\tt\small apedram@utexas.edu}}%
\thanks{$^{2}$ Department of Aerospace Engineering and Engineering Mechanics, University of Texas at Austin, Austin, TX, USA.
        {\tt\small ttanaka@utexas.edu}}%
}

\begin{document}

\maketitle
\thispagestyle{empty}
\pagestyle{empty}

\begin{abstract}
This paper considers the problem of closed-loop identification of linear scalar systems with Gaussian process noise, where the system input is determined by a deterministic state feedback policy. The regularized least-square estimate (LSE) algorithm is adopted, seeking to find the best estimate of unknown model parameters based on  noiseless measurements of the state.  We are interested in the fundamental limitation of the rate at which unknown parameters can be learned, in the sense of the D-optimality scalarization criterion subject to a quadratic control cost.  We first establish a novel connection between a closed-loop identification problem of interest and a channel coding problem involving an additive white Gaussian noise (AWGN) channel  with feedback and a certain structural constraint. 
Based on this connection, we show that the learning rate is fundamentally upper bounded by the capacity of the corresponding AWGN channel.  Although the optimal design of the feedback policy remains challenging, we derive conditions under which the upper bound is achieved. Finally, we show that the obtained upper bound implies that super-linear convergence is unattainable for any choice of the policy.   
\end{abstract}

\section{INTRODUCTION}
System identification (ID)
has been a subject of research interest broadly in control and computer science communities \cite{lennart1999system,sutton1998introduction}. 
Research efforts on the ID methods have generated a number of key results, such as the guaranteed consistency (convergence of estimates to true parameters) of the open-loop methods for uniform asymptotically stable systems under the persistently exciting inputs \cite{morgan1977uniform,maghenem2017strict,morgan1977stability}. 
Besides consistency, the rate at which the true values are recovered is another important aspect of system ID settings. In \cite{loria2004explicit}, an upper bound for the rate of convergence is provided for a specific class of ID schemes. It is also shown that for LTI systems an adaptive observer can achieve exponential convergence \cite{kreisselmeier1977adaptive}.

In this paper, we consider the problem of closed-loop ID of a linear stochastic system with quadratic cost. We seek to characterize feedback policies achieving high convergence rates with low input costs.  
The problem we study is generally known as the \textit{optimal input design} (see, e.g., \cite{bombois2011optimal} and the references therein) which has been extensively studied in the literature.
A large portion of input design methods for linear systems is in frequency domain and is based on the \textit{asymptotic theory}. This theory states that the  quadratic \textit{predictive estimator}, under some mild assumption, converges to the true parameters and   the covariance of estimation error decays linearly in number of samples $T$ as $\frac{1}{T} P$ in the limit of $T\rightarrow \infty$, where $P$ is referred to as asymptotic covariance matrix (see \cite[Chapter~9]{lennart1999system} for more details).

 It is shown for the open-loop ID of linear systems in \cite{lennart1999system} that the inverse of asymptotic covariance matrix is an affine function of the input power spectrum. Based on this affine relation, it is shown in \cite{jansson2005input} that by parametrizing the input power spectrum and imposing  the power constraint via Parseval's theorem, the input design problem can be formulated as an LMI. A deeper variance analysis for both open-loop and closed-loop ID of Box-Jenkins models is provided in \cite{bombois2005open}.
In the time domain approach, it is proposed in \cite{manchester2010input} to maximize, in some sense, the Fisher information matrix based on the recognition that the inverse of the Fisher information is a lower bound on the achievable covariance matrix of an unbiased estimator \cite{kailath2000linear}. 



Recently, non-asymptotic results are derived in \cite{simchowitz2018learning}  and \cite{dean2017sample} for ID of linear systems  based on non-stationary statistical analysis. In this approach, ID is performed via Gaussian excitation through independant \textit{rollouts} which are long sequences of inputs.  The upper bound for the estimation error of the parameters is then computed, using only the last measurement of each rollout.

 In the reinforcement learning literature, it is common to analyze the rate at which the optimal policy is reconstructed in terms of \textit{regret}. Regret is defined to be the performance deviation from the controller designed based on the true value of the parameters. For instance, \cite{abbasi2011regret} proposed to use the optimism in the face of uncertainty (OFU) principle, by propagating the confidence ellipsoid on the true parameters. 
The exponential dependence of regret on the order of the system was later reduced to linear dependency under further sparsity constraints on the dynamics in \cite{ibrahimi2012efficient}. A Thompson sampling-based learning algorithm with dynamic episodes was proposed in \cite{ouyang2017learning}, showing that the regret up to time $T$ is bounded by $O(\sqrt{T})$ . 

To gain further insights on the best possible learning rate, in this paper we introduce a new perspective on the closed-loop ID problems through the lens of feedback channel coding theory. To relate the system ID problem with information theory, we adopt the D-optimality scalarization metric instead of regret as the learning performance. We first observe that the system ID problems with input cost and the communication over noisy channels with input power constraint  share the same spirit of information gain maximization.
Inspired by this observation, our analysis builds a bridge between system ID and channel coding theories.

\subsection{Contributions}
We uncover a close kinship between  closed-loop ID schemes and communications schemes over the channels with feedback. This connection enables us to use the well established information-theoretic tools to provide non-stationary results, such as theoretical upper-bounds for the rate of convergence, and to specify the characteristics of optimal ID polices. More precisely, this paper makes the following contributions:
\begin{itemize}
    \item 
    We show the equivalence between the system ID problem with the D-optimality criterion and a channel coding problem  under the existence of noiseless feedback. To the best of our knowledge, this equivalence has not been established in the existing literature. 
    \item  A correspondence between the rate of convergence and the feedback channel capacity (the temporal average of directed information from channel input to its output) is made. Based on this analogy, a relation between the excitation cost and the upper-bound of  the achievable reduction of estimation error within a finite number of ID steps, which is a non-stationary analysis, is provided. This upper bound shows that the convergence is at most linear. 
    \item  Although the existence of an ID scheme that achieves the provided upper-bound is not known currently, the conditions under which the derived upper bound is achieved are discussed. More precisely,  we show that the upper bound is tight if and only if the state sequence is temporally independent and the input distribution of the corresponding channel matches the capacity achieving distribution.  
\end{itemize}

\subsection{Notation and Convention} 
Random variables and their realizations are denoted by upper-case and lower-case symbols respectively. The notation $X^t = (X_0,  X_1, \dots X_t)$ is used to denote the history of state $X$. $h(Y):=-\int p(y)\log p(y)dy$ is used to denote the differential entropy of random variable $Y$. For a Gaussian random variable $Y\sim \mathcal{N}(y,\Sigma_y)$, $h(Y)=\frac{1}{2}\log\det(2\pi e\Sigma_y)$. Mutual information between random variables $X$ and $Y$ is denoted by $I(X;Y):=h(X)-h(X|Y) = h(Y)-h(Y|X)$. 

\section{Problem Formulation}

Consider the discrete-time linear  system
\begin{align}
\label{dyn}
    X_{t+1}=A X_{t}+B U_{t}+W_t \quad X_0=0,
\end{align} 
where $X_t$ is a scalar-valued state and $W_t\sim \mathcal{N}(0, W)$ are i.i.d Gaussian random variables with variance $W$. We assume that the realization  of the parameter  $\Theta=[A \; B]^\intercal$ is not directly observable but the prior belief about $\Theta$ is known to be $\Theta \sim \mathcal{N}(\hat{\theta}_0, \Pi_0)$, where $\Pi_0 \succ 0$.

We formulate the optimal ID for the system defined above as finding the deterministic excitation policy $U_t=g_t(X^{t})$ optimizing a performance metric for a learning rate we introduce shortly. By defining $z_{t} :=[x_t \; u_t]^\intercal$, the state transition can be written as $x_{t+1}=z_{t}^\intercal \theta+ w_t$ for $ t=0, \dots, T-1$. We recursively compute the minimum variance unbiased estimator of the system's parameter $ \theta $, by exploiting the $l_2$-regularized LSE algorithm based on the full observation of state history $X^T$. Regularized estimation is selected to enforce the numerical stability of the algorithm and to exploit the initial knowledge of $\theta$. The best estimate can be characterized as \cite{kailath2000linear}:
\[\hat{\theta}_t=\argmin_{\theta} ||\theta-\hat{\theta}_0||_{\Pi_0^{-1}}^2+||\mathcal{X}_{t}-\mathcal{Z}_{t} \theta||_{\mathcal{W}^{-1}}^2,\] 
 where $\mathcal{Z}_t:=[z_0^\intercal; \dots; z_{t-1}^\intercal]$, $\mathcal{X}_t:=[x_1; \dots; x_{t}]$, and $\mathcal{W}:=WI$. We define $\Pi_t^{-1}:=\sum_{i=1}^{t} z_{i-1} W^{-1}z_{i-1}^\intercal+\Pi_0^{-1}$, which yields for all $t$ that $\Pi_{t}^{-1}=\Pi_{t-1}^{-1}+ z_{t-1} W^{-1}z_{t-1}^\intercal$, starting from $t=1$. By the matrix inversion lemma, the optimal estimate and its error covariance can be recursively computed as:
\begin{subequations}
 \label{LSE}
\begin{align}
\Pi_t=&\ \Pi_{t-1}-\frac{\Pi_{t-1} z_{t-1} z_{t-1}^\intercal \Pi_{t-1}}{W+z_{t-1}^\intercal \Pi_{t-1}z_{t-1}}\\ 
\hat{\theta}_{t}=&\ \hat{\theta}_{t-1}+ \frac{\Pi_{t-1} z_{t-1} }{W+z_{t-1}^\intercal \Pi_{t-1}z_{t-1}} (x_{t}- z_{t-1}^\intercal \hat{\theta}_{t-1}).
\end{align}
\end{subequations}

Note that this iteration corresponds to Kalman filter (KF) associated with the following system.
\begin{subequations}
\label{equi-sys}
 \begin{align}
 &\Theta_t=\Theta_{t-1}\\ \label{sensormodel}
 & X_t=Z_{t-1}^\intercal \Theta_t +W_{t-1},
 \end{align}
 \end{subequations}
 where $\hat{\theta}_t = \mathbb{E}[\Theta|X^t=x^t]$ is the estimator of $\Theta_t$ and $\Pi_t=\mathbb{E}[(\Theta-\hat{\theta}_t)(\Theta-\hat{\theta}_t)^\intercal| X^t=x^t]$ is the associated error covariance. 
 
Nonsingularity of $\mathbb {E}[Z_t Z_t^\intercal]$ guarantees the consistency (convergence to the true $\theta$) \cite{soderstrom1983instrumental}. However, it is worth noting that linear control policies $U_t=kX_t$ do not satisfy this nonsingularity condition and cannot be deployed for deterministic closed-loop ID of linear systems. We measure the performance of  the closed-loop ID process based on the following excitation cost and information utility. 
\subsection{Excitation Cost}
Inspired by the standard linear quadratic regulator (LQR), we adopt the following quadratic function of states and inputs:
\begin{equation}
\label{lqg}
J_T:= \frac{1}{T}( \sum_{t=0}^{T-1} \mathbb{E}[ q X^2_t+ r U^2_t]+q\mathbb{E}[X^2_T])
\end{equation}
as the excitation cost, where $q$  and $r$ are positive scalars.
 
 \subsection{Information Utility}
The information utility can be defined based on the so-called alphabetical design criteria \cite{chaloner1995bayesian}, which includes several scalarizations of error covariance such as A-optimality, D-optimality, E-optimality, and T-optimality for which the trace, the log-determinant, the largest eigenvalue, and the trace of the inverse of error covariance is adopted, respectively.

In this work, we select the D-optimality scalarization of the estimation error covariance with multiplicative factor $\frac{1}{2}$ as a metric of information utility. This metric is equivalent to the  information-theoretic quantity of \textit{entropy}. Specifically, the information utility obtained at time $t$ is defined as the entropy reduction:
\begin{align*}
F_t(X^t=x^t)=& \frac{1}{2}\log\det(\Pi_{t-1})-\frac{1}{2}\log\det(\Pi_{t}) \\
=&h(\Theta|X^{t-1}=x^{t-1})-h(\Theta|X^{t}=x^{t}).
\end{align*}
The accumulated utility up to time $T$ will be
\[L_T(X^T=x^T):=\sum_{t=1}^T F_t= \frac{1}{2}\log\det(\Pi_{0})-\frac{1}{2}\log\det(\Pi_{T}).\]
We stress that the information utilities $F_t(X^t)$ and $L_T(X^T)$ are in fact random variables, since $\Pi_t$ is computed based on the realizations of  $Z^t$ (or equivalently $X^t$) in (\ref{LSE}). If we consider the expectation of information utility as $\bar{F_t}:=\mathbb{E}[F_t]$, we have
\begin{align*}
    \bar{F_t} &= \int_{\Omega} F_t(X^t=x^t) \; dP_{x^t} \\
    &= \int_{\Omega} \big(h(\Theta|X^{t-1}=x^{t-1})-h(\Theta|X^{t}=x^{t})\big) \; dP_{x^t}\\
    &= h(\Theta|X^{t})-h(\Theta|X^{t-1}).
\end{align*}
The expectation of the accumulated information utility $\bar{L}_T:=\mathbb{E}[L_T]$ is equal to the mutual information $I(\Theta, X^T)$:
\begin{align}
\nonumber
\bar{L}_T &= \sum_{t=1} ^T \bar{F}_t = \sum_{t=1} ^T \big( h(\Theta|X^{t})-h(\Theta|X^{t-1}) \big) \\ \nonumber
&= h(\Theta)-h(\Theta|X^T)=I(\Theta;X^T).
\end{align}

Mutual information $\bar{L}_T=I(\Theta;X^T)$ provides a metric quantifying how much the uncertainty of $\Theta$ i.e., $h(\Theta)$ decreases in expectation after observing the realizations of $X^T$. Based on these interpretations, we respectively refer to $\bar{F}_t$, $\bar{L}_T$, $\frac{\bar{L}_T}{T}$ as information gain, cumulative information gain, and the rate of convergence in the following sections. 

\subsection{Characterization of Optimal Policy}
Using the aforementioned measures of convergence and control performance, the main optimization problem considered in this paper can be formulated as:
\begin{equation}
\begin{split}
\label{mainprob}
\max_{\{U_t=g_t (X^{t})\}_{t=0}^{T-1}} &I(\Theta; X^{T})\\
\text{s.t.} \quad  
 &\frac{1}{T}(\sum_{t=0}^{T-1} \mathbb{E}[q X^2_t+r U^2_t]+q\mathbb{E}[X^2_T]) \leq \gamma\\
&X_{t+1}=A X_{t}+B U_{t}+W_t.
\end{split}
\end{equation}
We are also interested in the optimal time invariant policy $U_t=g(X^t)$ for infinite horizon counterpart of the problem (\ref{mainprob}) as 
\begin{equation}
\begin{split}
\label{main-inf}
\max_{U_t=g_t (X^{t})}& \liminf_{T\rightarrow \infty}\frac{1}{T}I(\Theta; X^{T})\\
\text{s.t.}  \quad & 
 \limsup_{T\rightarrow \infty}\frac{1}{T} (\sum_{t=0}^{T-1} \mathbb{E}[q X^2_t+r U^2_t]+q\mathbb{E}[X^2_T]) \leq \gamma\\
&X_{t+1}=A X_{t}+B U_{t}+W_t.
\end{split}
\end{equation}

\section{Proposed Approach }
The problem of maximizing mutual information with power constraints similar to  problems (\ref{mainprob}) and (\ref{main-inf}) has a long history in the study of communication over noisy channels. This similarity motivates us to bridge system ID and information theory by finding  channel coding problems analogous to (\ref{mainprob}) and (\ref{main-inf}). 

In this paper, in lieu of direct analysis, we model the system ID algorithm as a communication scheme over a Gaussian channel with feedback. To the best of our knowledge, this is the first time that such a connection is shown explicitly and the equivalence between the rate of convergence of the system ID and directed information between channel input and output is demonstrated.

We extend the existing results in the context of channel coding and investigate their implications such as theoretical bounds for the rate of convergence and the required conditions to achieve such bounds.

\subsection {Communication over a Noisy Channel with Feedback}
In this section, we review basic results for the problem of communication over noisy channels with feedback, which has been studied extensively in the network information theory literature \cite{el2011network}.
Consider the problem of transmitting a message $M$ over an additive white Gaussian noise (AWGN) channel with noiseless feedback depicted in Figure~\ref{fig:3}. 
The noise $W_t$ is drawn i.i.d from a Gaussian distribution $W_t\sim \mathcal{N}(0, W)$ and it is independent of input signal $Y^t$. The output of the channel is $X_t=Y_t+W_t$.

\begin{figure}[thpb]
    \centering
    \includegraphics[width=\columnwidth]{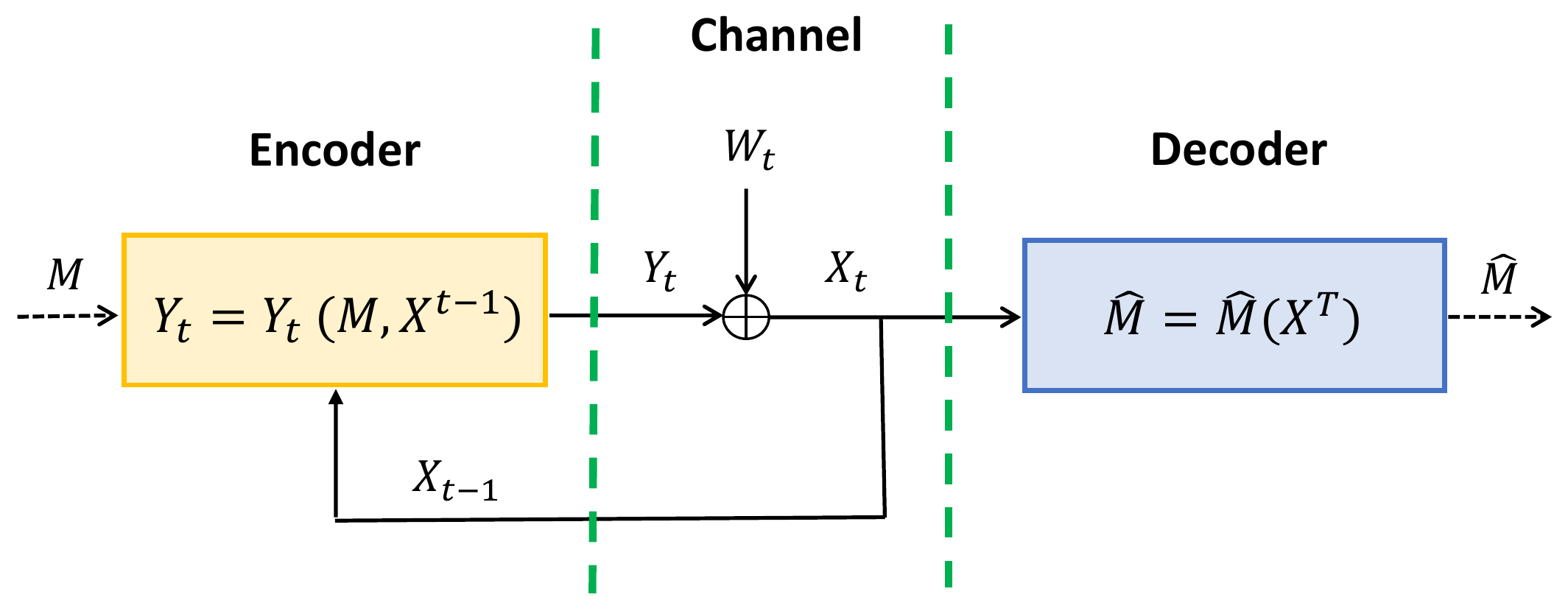}
    \caption{General Gaussian channel with feedback.}
    \label{fig:3}
\end{figure}

The capacity of the AWGN channel with feedback, denoted by $C_{FB}(P)$, is defined as the largest $R\in \mathbb{R}_{+}$ for which there exists an encoder-decoder pair such that the associated sequence of inputs $\{y_t(M,X^{t-1})\}_{t=1}^T$ for every message $M \in \{1, 2, \dots, 2^{TR}\}$ satisfies the power input constraint
\[ \sum_{t=1}^T\mathbb{E}_W[y^2_t(M,X^{t-1})]\leq TP,\]
 and the message can be decoded by the decoder after $T$ channel uses with diminishing probability of decoding error, i.e., $\limsup_{T\rightarrow \infty} P(\hat{M} \neq M)=0$. The reader is referred to \cite{cover2012elements,el2011network} for further discussion. The next theorem exhibits a connection between the feedback channel capacity $C_{FB}(P)$ and the directed information $I(Y^T\rightarrow X^T)=\sum_{t=1}^T I(Y^t;X_t|X^{t-1})$ for Markov channels \cite{cover2012elements} with feedback.
 
 \begin{theorem}
\label{theo-fbc}
For Markov channels with feedback, channel capacity is equal to the limsup of the temporal average of directed information from the channel input $Y^T$ to its output $X^T$ maximized over all causally conditioned
distributions $P(Y^T||X^{T-1}):=\Pi_{t=1}^T P(Y_t|X^{t-1})$ satisfying the power constraint i.e.,   
 \begin{equation}
 \label{fbc}
 \begin{split}
     C_{FB}(P)= \limsup_{T\rightarrow\infty} \!\!\!\! \max_{\substack{P(Y^T||X^{T-1})\\ \text{s.t.}~
     \frac{1}{T}\sum_{t=1}^T\mathbb{E}[Y^2_t]\leq P}} \!\!\!\!
     \frac{1}{T} I(Y^T\rightarrow X^T).
 \end{split}
 \end{equation}
 \end{theorem}

 \begin{proof}
    See \cite{tatikonda2008capacity} for proof. Similar results are derived for the finite state channels without inter-symbol interference in \cite{permuter2009finite} and for  stationary nonanticipatory channels in \cite{kim2008coding}.
\end{proof}

It is well known \cite{shannon1956zero} that feedback does not increase the capacity of memory-less channels (including AWGN channels). Consequently, we have $C_{FB}(P)=C(P)$, where $C(P):=\frac{1}{2}\log(1+\frac{P}{W})$ is the capacity of AWGN channel without feedback \cite{cover2012elements}. Nevertheless, feedback helps to simplify the coding scheme and increases the achievable error exponent \cite{schalkwijk1966coding}.

\section{Main Results}

In this section, we construct a channel coding problem which is equivalent to the online parameter ID problems (\ref{mainprob}) and (\ref{main-inf}). We then demonstrate that cumulative information gain  and the directed information from the input to the output of the equivalent communication scheme  are equal. Based on this equivalence and Theorem~\ref{theo-fbc}, it is shown that the asymptotic rate of convergence i.e. the solution to (\ref{main-inf}) is upper-bounded by the capacity  of the corresponding channel $C(P)$. Even for case (\ref{mainprob}) with the finite number of ID steps where Theorem~\ref{theo-fbc} is not directly applicable, we show that the step-wise information gain is upper bounded by  $C(P)$ and that the value of (\ref{mainprob}) is upper bounded by $TC(P)$.

We provide necessary and sufficient conditions to achieve these upper limits. More precisely, an excitation policy $U_t=g_t(X^t)$ is shown to be optimal if and only if it is matched with the capacity achieving distribution for the corresponding channel and $X_t$ is an independent random process. \footnote{It is currently not known if there exists a deterministic feedback policy $U_t=g_t(X^t)$ that meets this optimality condition.} Finally, we prove that input power of the equivalent channel $P$ is bounded and finite for every policy with finite control cost $J$. Therefore, cumulative information gain $\bar{L}_T$ is upper bounded by the linear function of $TC(P)$ meaning that $\frac{\bar{L}_T}{T}=\mathcal{O}(1)$ and super-linear convergence is impossible.

\subsection{Equivalent Channel Formulation}
A system ID scheme with a feedback policy $U_t=g_t(X^t)$ can be modeled  as communication system over a Gaussian channel with noiseless, one-step delayed feedback as depicted in Figure~\ref{fig:1}. The sensor model (\ref{equi-sys}) is interpreted as an AWGN channel over which the message $\theta$ is communicated. The decoder tries to estimate $\theta$ based on the sequence $X^T$ of the channel outputs using (\ref{LSE}).


The following theorem states that the aforementioned measure of convergence is equivalent to directed information between the channel input and output. This theorem also indicates a close connection between the parameter ID problems (\ref{mainprob}) and (\ref{main-inf}) and the feedback channel coding problem (\ref{fbc}).
\begin{figure}[thpb]
    \centering
    \includegraphics[width=\columnwidth]{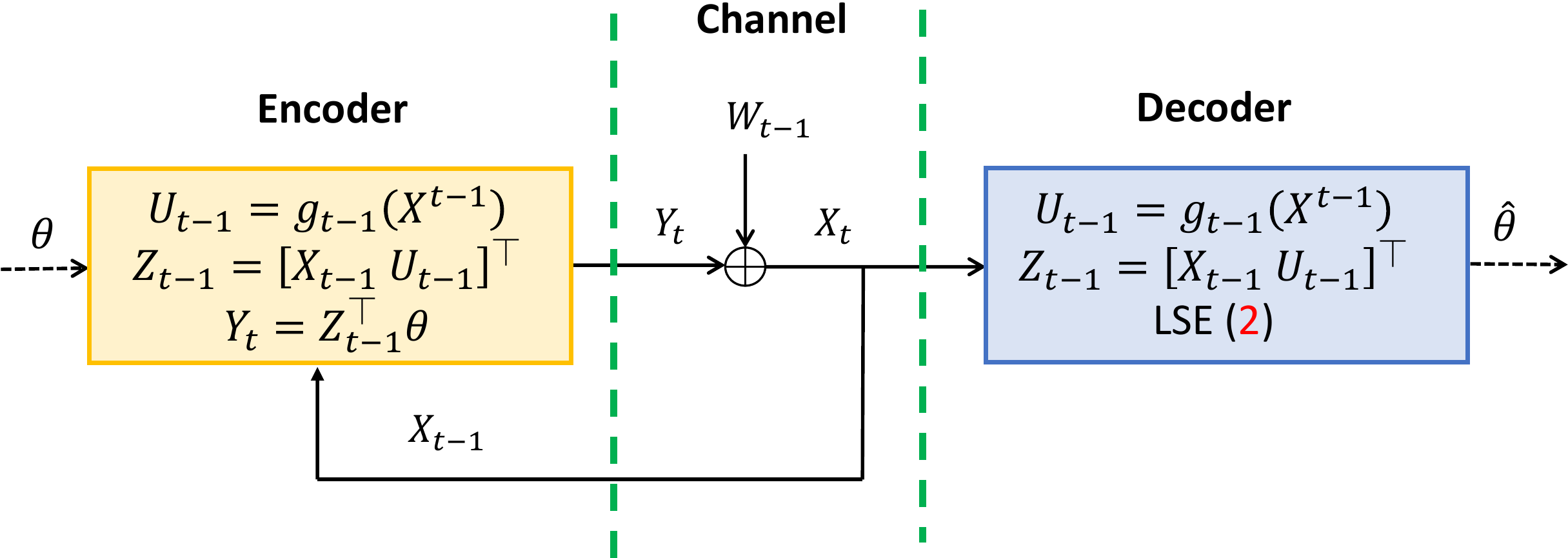}
    \caption{Channel representation of system ID with deterministic policy.}
    \label{fig:1}
\end{figure}

 \begin{theorem}
 \label{theo2}
For any deterministic control policy $U_t=g_t(X^{t})$, we have
\[ I(\Theta; X^T)= I(Y^T \rightarrow X^T),\]
and $\mathbb{E}[\det(\Pi_T)]\geq \det(\Pi_0)2^{-2I(Y^T\rightarrow X^T)}$.
 \end{theorem}
\begin{proof}
By the chain rule, we have
\begin{subequations}
\begin{align}
I(\Theta;X^T)&=\sum_{t=1}^{T} I(\Theta; X_t| X^{t-1})\\  \label{ach1}
&=\sum_{t=1}^{T}\left( h( X_t| X^{t-1})-h(X_t|\Theta, X^{t-1})\right)\\ \label{ach2}
&=\sum_{t=1}^{T}\left( h(X_t| X^{t-1})-h(X_t|Y^t, X^{t-1})\right)\\
&=\sum_{t=1}^{T} I(Y^t;X_t|X^{t-1})= \!I(Y^T\rightarrow X^T), 
\end{align}
\end{subequations}
where from \eqref{ach1} to \eqref{ach2} we used the fact that the coding scheme is a 1-1 causal mapping from $Y^t$ to $(\Theta,X^{t-1})$. Additionally, $\frac{1}{2}\log(\mathbb{E}[\det(\Pi_T)]) \geq \frac{1}{2}\mathbb{E}[\log\det(\Pi_T)]=\frac{1}{2}\log\det(\Pi_0)-I(Y^T\rightarrow X^T)$, which yields $\mathbb{E}[\det(\Pi_T)]\geq \det(\Pi_0)2^{-2I(Y^T\rightarrow X^T)}$.
\end{proof}

\subsection{Upper Bound Analysis}

The equivalent channel representation of the system ID problem (Figure~\ref{fig:1}) is a specific example of general communication schemes over the additive Gaussian channel depicted in Figure~\ref{fig:3}. Therefore, the maximum directed information achievable by general feedback scheme is an upper bound to the directed information achievable by system ID scheme. In particular, $C_{FB}(P)$ characterized by Theorem~\ref{theo-fbc} provides an upper bound to the asymptotic learning rate i.e., $\lim_{T\rightarrow \infty} \frac{\bar{L}_T}{T}$ or the solution to (\ref{main-inf}).

These asymptotic results may not be sufficient to analyze the non-stationary performance when only a finite number of ID steps are performed. In the next lemma, we derive an upper-bound for obtainable information gain in each step, i.e., an upper bound for (\ref{dir-pr}). Note that
\begin{align*}
\bar{F}_t=&I(\Theta; X^{t})-I(\Theta; X^{t-1})\\
=&I(Y^t\rightarrow X^t)-I(Y^{t-1}\rightarrow X^{t-1})=I(Y^t;X_t|X^{t-1}).
\end{align*}
Assuming a step-wise energy constraint, $\mathbb{E}[Y^2_t]\leq P_t$, we intend to solve
\begin{equation}
\begin{split}
\label{dir-pr}
  \max_{P(Y_t|X^t)}\  & I(Y^t;X_t|X^{t-1})\\
  \text{s.t.} \quad &\mathbb{E}[Y^2_t]\leq P_t.
  \end{split}
\end{equation}

\begin{lemma}
For any coding-decoding scheme that  satisfies $\mathbb{E}[Y^2_t]\leq P_t$, the increase in the directed information i.e.,  $I(Y^t;X_t|X^{t-1})$ is bounded from above by $C(P_t)$. Consequently,  \[\bar{L}_T=I(\Theta;X^T)=I(Y^T\rightarrow X^T) \leq \sum_{t=1}^T C(P_t).\]
\begin{proof}
\begin{subequations}
\begin{align}
\nonumber
    &I(Y^t;X_t|X^{t-1})\\ \label{lem1a}
    &=h(X_t| X^{t-1})-h(X_t|Y^t, X^{t-1}) \\ \label{lem1b}
    & \leq h(X_t)-h(W_{t-1})\\ \label{lem1c}
    & = h(Y_t+W_{t-1})-h(W_{t-1})\\ \label{lem1d}
    & \leq \frac{1}{2} \log (1+\frac{P_t}{W})=C(P_t),
\end{align}
\end{subequations}
where from (\ref{lem1a}) to (\ref{lem1b}) we  used the fact that conditioning won't increase the entropy i.e. $h(X_t|X^{t-1})\leq h(X_t)$. The inequality in (\ref{lem1d}) comes from the non-Gaussianity of $Y_t$ and the power constraint in (\ref{dir-pr}). Notice that $Y_t$ and $W_{t-1}$ are independent. 
\end{proof}
\end{lemma}
 
\begin{corollary}
\label{coro}
For the problem with power constraint
\begin{equation}
\begin{split}
\label{dir-pr-s}
  \max_{P(Y^T||X^T)}\  & I(Y^T\rightarrow X^{T})\\
  \text{s.t.} \quad &\frac{1}{T}\sum_{t=1}^T\mathbb{E}[Y^2_t]\leq P,
  \end{split}
\end{equation}
we have $\bar{L}_T= \sum_{t=1}^T \bar{F}_t \leq \sum_{t=1}^T C(P_t) \leq T C(P)$, where $P_t$ denotes channel input power at time $t$. The final inequality comes from the fact that $C(.)$ is concave and $\frac{1}{T}\sum_{t=1}^T P_t=P$. This corollary extends the asymptotic result mentioned in Theorem~\ref{theo-fbc} to cases with finite $T$.
\end{corollary} 

Corollary~\ref{coro} states that the solution to (\ref{mainprob}) is bounded from above by $TC(P)$. In the following sections, we discuss the conditions under which this upper bound is achievable and build the connection between the excitation cost and the input power of the equivalent channel.  

\subsection{Tightness of the Upper Bound}
If we assume $Y_t$ is Gaussian and we make full use of power budget i.e. $\mathbb{E}[Y^2_t]= P_t$, the inequality (\ref{lem1d}) holds with equality. The inequality in (\ref{lem1b}) is also \textit{tight} if $X_t$ and $X^{t-1}$ are independent. This tightness, presuming the (\ref{lem1d}) holds with equality, 
proves that a coding scheme achieves the maximum possible increase of the information gain $C(P_t)$ if and only if $X_t$ and $X^{t-1}$ are independent. This supports the idea that the excitation strategy is optimal (most informative) if it results in statistically independent observations $X_t$.

In the field of information theory, the fact that a coding-decoding scheme which achieves the channel capacity will produce a sequence of temporally independent channel output is not a new result. For example, it is shown in \cite{el2011network} that the well-known Schalkwijk-Kailath (SK) coding scheme induces independent  outputs. 

In our ID setting, the upper bound can be achieved if and only if $X_t$ is an independent process and the input distribution in the equivalent channel matches with one of the capacity achieving distributions like the SK scheme. Unfortunately, our equivalent coding scheme has a restricted structure and no choice of excitation policy $U_t=g_t(X^t)$ recovers SK scheme. Therefore, the SK scheme cannot be directly implemented in our ID setting.

\subsection{Infeasibility of Super-linear Convergence}
In previous sections, the provided bounds are accounting for the channel input power $P$ and not the excitation cost $J$. Establishing an explicit relation between the excitation cost and the channel input power for general excitation policy $U_t=g_t(X^t)$ is nontrivial. However, it is shown in the next theorem that $P$ is finite under any excitation scheme that yields a finite $J$.
\begin{theorem}
\label{theo-bound}
For any excitation policy 
\begin{align}
P_{T} \leq \frac{2}{q} J_T+2W.
\end{align}
\end{theorem}
\begin{proof}    
\begin{subequations}
\begin{align}
\label{bound-1}
P_T &=\frac{1}{T}\sum_{t=1}^T \mathbb{E}[Y_t^2]=\frac{1}{T}\sum_{t=1}^T \mathbb{E}[(X_{t}-W_{t-1})^2]\\ \label{bound-2}
&\leq \frac{1}{T}\sum_{t=1}^T \mathbb{E}[2 X^2_{t}+2 W^2_{t-1}] 
\leq \frac{2}{q} J_T+2W.
\end{align}
\end{subequations}
From (\ref{bound-1}) to (\ref{bound-2}), we have used Cauchy-Schwarz inequality. Last inequality (\ref{bound-2}) is trivial result of  (\ref{lqg}).
\end{proof}

Theorem~\ref{theo-bound} implies that for any deterministic closed-loop system ID scheme with bounded excitation cost, $C(P)$ is also bounded and thus $\bar{L}_T=\mathcal{O}(T)$, meaning that the scheme can at most achieve a linear convergence. To the best of our knowledge, the infeasibility of super-linear convergence was not shown previously for finite number of ID steps $T$. Theorem~\ref{theo-bound} completes our analysis and establishes an explicit relation between the excitation cost and the information gain as $\bar{L}_T\leq TC(\frac{2}{q} J_T+2W)$.

\section{Numerical Demonstration}
In this section, we consider an example of a linear system with initial belief as $\Theta \sim \mathcal{N}([0\;0]^\intercal, I_{2})$ and $W=0.1$ in infinite horizon limit, where the true parameters are $\theta=[a \; b]^\intercal=[0.9 \; 1]^\intercal$. The excitation policy is adopted as $U_t=kX_{t}(1+\sin(X_t))$ for different stable values of linear feedback gain $k\in(-1.4,-0.9)$, where the stability can be demonstrated by Popov Criterion \cite{khalil2002nonlinear}. Based on the data from Monte Carlo simulation of the system averaged for $1000$ realizations, an empirical comparison between the upper bound for the asymptotic convergence rate and actual convergence rate as functions of channel input power $P$ is provided in Fig.~\ref{fig-info}. 
\begin{figure}[thpb]
      \centering
      \includegraphics[scale=0.55]{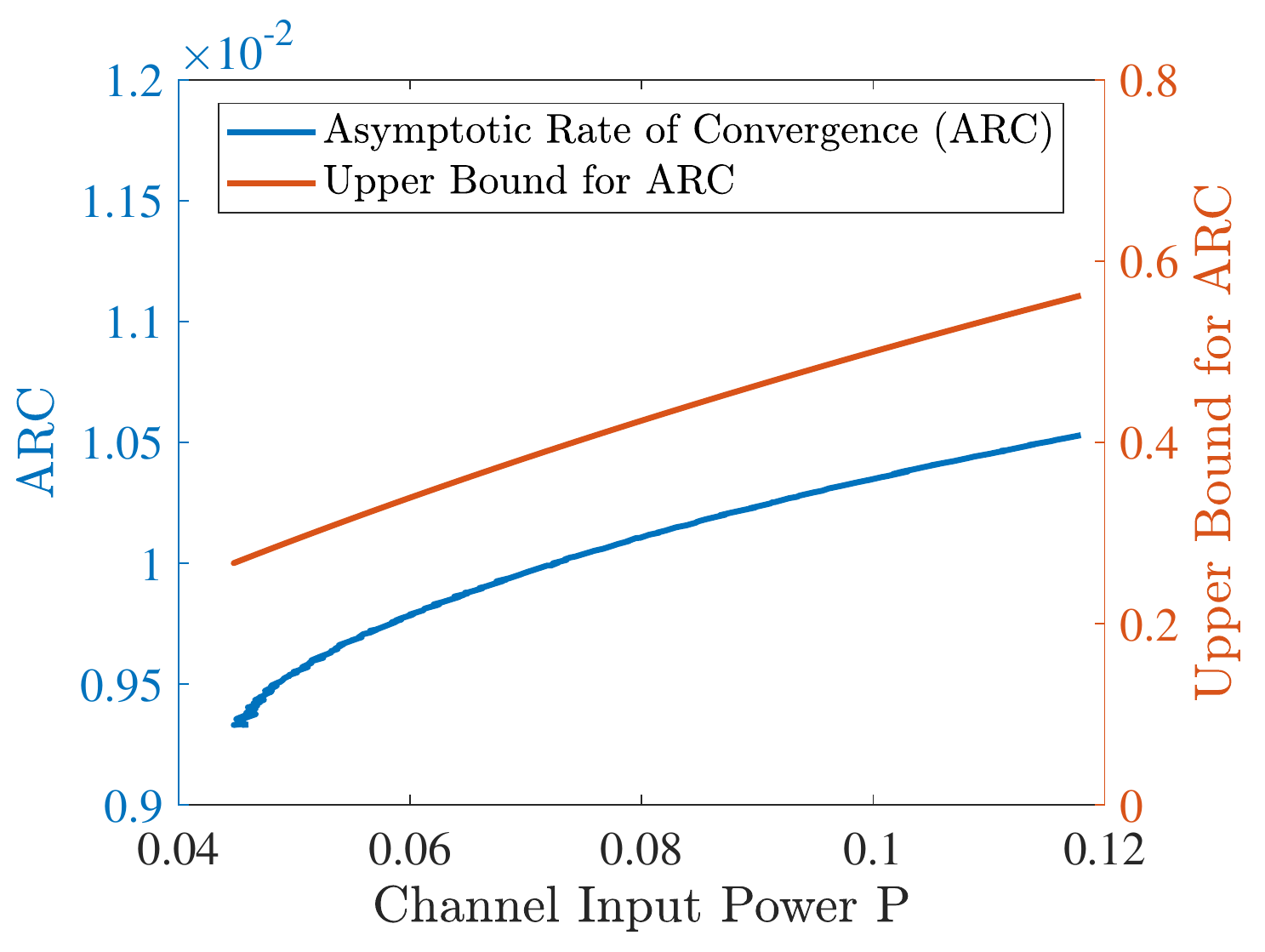}
      \caption{Actual asymptotic convergence rate versus the upper bound of the asymptotic convergence rate.}
      \label{fig-info}
   \end{figure} 
As demonstrated in Fig.~\ref{fig-info}, the  naive selection of excitation policy will result in poor performance for our choice of $U_t$, there is a noticeable gap between the provided upper bound and the actual rate of convergence.

\section{Conclusion and Future Work}
We studied a fundamental limitation of system ID schemes of linear systems under quadratic control costs via deterministic excitation policies. We show that under specific choice of information utility, the closed-loop ID problem can be modeled as Gaussian channel coding problem with feedback. With this connection, we showed that the rate of convergence (information utility obtained per ID step) is bounded from above by capacity of the equivalent channel. This implies that the convergence is always sublinear or linear at best. Finally, it was shown that a feedback policy achieves the provided upper-bound if and only if the input distribution to the corresponding channel matches the capacity achieving distribution and the characteristics of the capacity achieving schemes like independence of output process were discussed.

Our future research will focus on discovering, or examining the existence of, policies that satisfy this optimality condition. If such optimal policy does not exist, the main question will be how to narrow the gap between the achievable convergence rate and its upper bound. Generalization of the modeling and the results over the space of non-deterministic (stochastic) policies is another interesting future direction. We will also investigate the relationship between information utility metric considered in this paper and regret.     

\addtolength{\textheight}{-12cm}   






\bibliographystyle{IEEEtran}
\bibliography{ref}

\end{document}